\documentclass{aims}

\usepackage{amsmath}
\usepackage{color}
\usepackage[colorlinks=true]{hyperref}
\hypersetup{urlcolor=blue, citecolor=red}

\def\currentvolume{X}
\def\currentissue{X}
\def\currentyear{200X}
\def\currentmonth{XX}
\def\ppages{X--XX}
\def\DOI{10.3934/xx.xx.xx.xx}

\usepackage{amssymb}

\usepackage{amsthm}
\usepackage{mathrsfs}

\usepackage{pgfplots}
\usepackage{tikz}
\usepackage{tabulary}
\usetikzlibrary{cd}
\usetikzlibrary{calc,tikzmark}

\newtheorem{thm}{Theorem}[section]

\newtheorem{lem}[thm]{Lemma}
\newtheorem{prop}[thm]{Proposition}
\newtheorem{conj}{Conjecture}
\theoremstyle{definition}
\newtheorem{defn}[thm]{Definition}
\newtheorem{rmk}{Remark}
\newtheorem*{notation}{Notation}

\title{The Weight Distribution of Quasi-Quadratic Residue Codes}

\author[Nigel Boston and Jing Hao]{}
 \email{nboston@wisc.edu}
 \email{jing.hao@wisc.edu}

\begin{document}
\maketitle
\centerline{\scshape Nigel Boston}
\medskip
{\footnotesize
 \centerline{Department of Mathematics,}
   \centerline{Department of Electric and Computer Engineering,}
   \centerline{University of Wisconsin-Madison}
   \centerline{WI 53706, United States}
} 

\medskip

\centerline{\scshape Jing Hao}
\medskip
{\footnotesize
 \centerline{Department of Mathematics,}
   \centerline{University of Wisconsin-Madison}
   \centerline{WI 53706, United States}
}

\bigskip

\section{Introduction}
\label{sec:introduction}

Quasi-quadratic residue codes (QQR codes) are a family of binary linear codes.  They were first introduced by Bazzi and Mitter\cite{Bazzi2003} as a quasi-cyclic code.  Their work set foundations for the study of QQR codes. They discovered the relation between weights of a QQR code and number of points on hyperelliptic curves.  Joyner\cite{Joyner2008} built upon this relation, and revealed the link between the QQR code and the famous Goppa's Conjecture in coding theory.

We are interested in these codes mainly for two reasons:  Firstly, they have close relations with hyperelliptic curves and Goppa's Conjecture, and serve as a strong tool in studying those objects.  Secondly, they are very good codes.  Computational results show they have large minimum distances when $p\equiv 3 \pmod 8$.  

QQR codes are similar to quadratic residue codes.  They are asymptotically rate half codes (exactly rate half when $p\equiv 3 \pmod 4$). Also, as we will show, $PSL_2(p)$ acts as automorphisms of the extended QQR codes in a similar way as of the extended quadratic residue codes.  Furthermore, when $p \equiv 7 \pmod 8$, we will show that the QQR code is equivalent to the even subcode of the corresponding quadratic residue code direct sum with itself, and therefore their weight enumerators have close relations.

We will utilize the result that $PSL_2(p)$ acts on these codes to prove a new discovery about their weight polynomials, i.e. they are divisible by $(x^2 + y^2)^{d-1}$, where $d$ is the corresponding minimum distance.  The proof uses shadows of codes, a powerful tool to study weight polynomials.  We also apply this idea to quadratic residue codes, and prove that their weight polynomials are divisible by $(x + y)^d$, with $d$ being the minimum distance.

These results impose strong conditions on the weight polynomials of quadratic residue codes and QQR codes.  Combining the divisibility result and Gleason's Theorem, we can derive an efficient algorithm to compute the weight polynomials of QQR codes.  We also use these results to correct the existing computational results for the weight polynomials of quadratic residue codes that were originally posted on \cite{Tjhai}.

We also answer in the negative the question posted by Joyner\cite{Joyner2008} asking whether QQR codes satisfy Riemann hypothesis.

On the other hand, the weight of their codewords can be expressed in terms of the number of points on corresponding hyperelliptic curves over finite fields.  As it is usually easier to study linear codes, this provides a way of studying point distributions of hyperelliptic curves.  We will implement this idea to prove a variant of a result of Larsen\cite{Larsen2008} on asymptotic normal distribution of numbers of points on hyperelliptic curves.

\section{Construction and properties}
\label{sec:construction_and_properties}

We first give constructions and introduce properties of QQR codes.  We also include an introduction to quadratic residue codes, as these will be used in later sections.

{\bf Throughout the entire paper, if not stated otherwise, $p$ is a prime satisfying $\bf{p\equiv 3 \pmod 4}$.}

Let $S$ be a subset of $ \mathbb{F}_p$.  We can assign to $S$ a polynomial in $ \mathbb{F}_2[x]/(x^p - 1)$ by

\[
	r_S := \sum_{ a \in S} x^a
\]

Let $Q$ be the set of quadratic residues in $ \mathbb{F}_p$ and $N$ be the quadratic non-residues in $ \mathbb{F}_p$.  

\begin{defn}[Quadratic residue code]
	Let $p \equiv \pm 1 \pmod 8$, 
	\begin{align*}
		Q:= \{ (r_Qr_S)|S \subseteq \mathbb{F}_p\} \\
		N:= \{ (r_Nr_S)|S \subseteq \mathbb{F}_p\}
	\end{align*}
	are the \emph{quadratic residue codes} associated with $p$.  

	$Q$ is equivalent to $N$ since $r_N$ can be obtained from $r_Q$ via the permutation
	\[
		y \mapsto \rho y, y=0, \cdots, p-1
	\]
	where $\rho$ is a primitive element of $ \mathbb{F}_p$.
\end{defn}

\begin{notation}
	Note that with abuse of notation, we use $Q$ to denote both the quadratic residues and the code generated by $r_Q$.  It should be clear in the context what we are referring to.  Similarly for $N$.

	The generating matrices for $Q$ and $N$ are circulant matrices, and we denote them as $G_Q$ and $G_N$ respectively.
\end{notation}

\begin{defn}[Quasi-quadratic residue code] 
	For an odd prime $p$,
	\begin{align*}
		C:=\{(r_Qr_S,r_Nr_S)|S \subset \mathbb{F}_p\}		\end{align*}
	is called the \emph{quasi-quadratic residue code} associated with $p$.  $(r_Qr_S,r_Nr_S)$ is identified with an element in $ \mathbb{F}_2^{2p}$ in the usual way.
\end{defn}

\begin{notation}
	By ``the corresponding quadratic residue code'' to a QQR code, we mean the quadratic residue code associated with the same $p$.  Similarly for ``the corresponding QQR code'' to a quadratic residue code, we mean the QQR code associated with the same $p$.
\end{notation}

If we write 
\begin{align*}
	r_Q &= \sum_0^{p-1} a_i x^i \\
	r_N &= \sum_0^{p-1} b_i x^i
\end{align*}
then the generating matrix for the QQR code can be written as 

\[
	G:=
	\left[
		\begin{array}{c c c c c | c c c c c}
			a_0 & a_1 & a_2 & \cdots & a_{p-1} & b_0 & b_1 & b_2 & \cdots & b_{p-1} \\
			a_{p-1} & a_0 & a_1 & \cdots & a_{p-2} & b_{p-1} & b_0 & b_1 & \cdots & b_{p-2} \\
			\vdots & & \ddots & & \vdots & \vdots & & \ddots & & \vdots  \\
			a_2 & a_3 & a_4 & \cdots & a_1 & b_2 & b_3 & b_4 & \cdots & b_1 \\ 
			a_1 & a_2 & a_3 & \cdots & a_0 & b_1 & b_2 & b_3 & \cdots & b_0
		\end{array}
	\right]
\]

\bigskip

This generating matrix is double circulant.  Clearly $G=[G_Q |G_N]$.

We list some known properties of QQR codes below.  Interested readers can check \cite{Joyner2008} for more information.

QQR codes are even weight codes.  They have length $2p$ and dimension $p$.  QQR codes are self-dual.  

\begin{prop}[Gaborit\cite{Gaborit2001}]
	When $p \equiv 3 \pmod 8$, $G_Q^2 = G_Q^{-1} =G_N$.  Equivalently, $r_Q^2 = r_Q^{-1} = r_N$.
	\label{prop6}
\end{prop}
This proposition is a result of Perron's Theorem\cite{Macwilliams1977} on quadratic residues.

\begin{prop}[Bazzi, Mitter\cite{Bazzi2006a}]
	When $p\equiv 3 \pmod 8$, the QQR code also has a standard double circulant form, i.e. its generating matrix can be written as $[I|G_Q]$.
\end{prop}

\begin{proof}
	By Proposition \ref{prop6}, $G_Q$ is invertible and 
	\[
		G_Q[I|G_Q] = [G_Q|G_N]
	\]
	Therefore $[I|G_Q]$ is also a generating matrix for the QQR code.
\end{proof}

When Bazzi and Mitter first introduced QQR codes, their generating matrices were given in the form $[I|G_Q]$. They gave a proof that these codes also have generating matrices in the form $[G_Q|G_N]$ when $p \equiv 3 \pmod 8$.  Most references on double circulant codes study codes whose generating matrices are in the form $[I|A]$, where $A$ is cyclic, such as Karlin's original paper on double circulant codes\cite{Karlin1969}.  

Joyner defined QQR codes using $[G_Q|G_N]$ as their generating matrices, and we will also use this version in our paper.  Note that when $p \equiv 7 \pmod 8$, the code generated by $[I|G_Q]$ is not equivalent to the code generated by $[G_Q|G_N]$, and therefore these two definitions are not the same.

From the definitions of quadratic residue code and QQR code, it is obvious that when $p \equiv 7 \pmod 8$, by only taking the first $p$ bits of a QQR code, we can obtain the corresponding quadratic residue code.  Moreover, we can show a stronger connection between quadratic residue codes and QQR codes in this case.

\begin{notation}
We denote $d(C)$ as the function that outputs the minimum distance of a code $C$.
\end{notation}

\begin{prop}
	When $p \equiv 7 \pmod 8$, let $C$ be the QQR code associated with $p$, and let $Q$ and $N$ be the corresponding quadratic residue codes.  Then $C$ is the even subcode of $Q \oplus N$.

	Moreover, their minimum distances satisfy
	\[
		d(C) = d(Q) + 1
	\]
	\label{prop9}
\end{prop}

\begin{proof}
	Since $C = \{(r_Qr_S,r_Nr_S)|S \subseteq \mathbb{F}_p\}$, $Q \oplus N = \{ (r_Q r_{S_1}, r_N r_{S_2})|S_1,S_2 \subseteq \mathbb{F}_p\}$, we have $C \subseteq Q \oplus N$.

	Also $C$ is even, and therefore $C$ is a subcode of the even subcode of $Q \oplus N$.  $Q$ contains the all 1 codeword, and hence is not even.  Therefore $Q \oplus N$ is also not even, and its even subcode is of codimension 1, which is $\dim(Q) + \dim(N) - 1 = \frac{p+1}{2} + \frac{p+1}{2} - 1 = p$.

	Since $\dim(C) = p$, $C$ is equal to the even subcode of $Q \oplus N$.	

	For the last statement, as we will show in Proposition \ref{prop7}, $d(Q)$ is odd. Also, there exists a codeword $c'$ in $C$ with weight $d+1$\cite{Blahut2008}.  Therefore $d(C)\ge d+1$.  On the other hand, since $d+1$ is even, $c'\oplus 0$ is in the even subcode of $Q \oplus N$.  So $d(C) = d(Q) + 1$.
\end{proof}

\subsection{Hyperelliptic Curves and Goppa's Conjecture}
\label{sub:relation_with_hyperelliptic_curves_and_goppa_s_conjecture}

For a code $C[n,k,d]$, $R:=\frac{k}{n}$ is the information rate.  $\delta:=\frac{d}{n}$ is the relative minimum distance.  $\delta$ indicates the error-correcting ability of a code.  Ideally we want $R$ and $\delta$ both large, but Manin\cite{Manin1982} proved that for a fixed field $\mathbb{F}_q$, there exists a function $ \alpha_q(\delta)$ such that for a given $\delta$, there are infinitely many linear codes with rate approaching $R$ only for rates below $ \alpha_q(\delta)$.

Gilbert\cite{Gilbert1952} and Varshamov\cite{Varshamov1957} showed $ \alpha_q(\delta) \geq 1 - x \log_q (q-1) + x \log_q(x) + (1-x) \log_q(1-x)$.  When $q=2$, this is believed to be an equality by many people, known as the following conjecture.

\begin{conj}[Goppa's Conjecture]
	The Gilbert-Varshamov bound is tight in the binary case.
\end{conj}

QQR codes play an important role in the study of Goppa's Conjecture because of the following explicit relation with hyperelliptic curves.

\begin{notation}
	We denote by $wt(c)$ the function that outputs the weight of a codeword $c$.
\end{notation}

\begin{prop}[Joyner\cite{Joyner2008}]
	Let $C$ be a QQR code associated with $p$, and $c=(r_Qr_S,r_Nr_S) \in C$, where $S \subseteq \mathbb{F}_p$.

Define $f_S(x) := \prod_{a \in S} (x-a)$.  Let $X_S( \mathbb{F}_p)$ be the set of points on the hyperelliptic curve $y^2 = f_S(x)$ over $ \mathbb{F}_p$.	

	\begin{enumerate}
		\item If $|S|$ is even 
			\[
				wt(c) = 2p - |X_S(\mathbb{F}_p)|
			\]
		\item If $|S|$ is odd and $p \equiv 1 \pmod 4$
			\[
				wt(c) = 2p - |X_{S^c}(\mathbb{F}_p)|
			\]
		\item If $|S|$ is odd and $p \equiv 3 \pmod 4$
			\[
				wt(c) = |X_{S^c}(\mathbb{F}_p)|
			\]
	\end{enumerate}
	\label{prop-Joyner}
\end{prop}

\begin{rmk}
	By the points on hyperelliptic curves we mean affine points, not including the points at infinity.
\end{rmk}

This relation builds a connection between Goppa's Conjecture and hyperelliptic curves as in the following theorem.  Interested readers can find more details in \cite{Joyner2008}.

\begin{thm}[Joyner]
	Let $B(c,p)$ be the statement: For all subsets $S \subseteq \mathbb{F}_p$, $|X_S(\mathbb{F}_p)| \le c \cdot p$ holds.  If $B(1.77,p)$ is true for infinitely many primes with $p\equiv 1 \pmod 4$, then Goppa's Conjecture is false.
\end{thm}

Not only do QQR codes play an important role in this connection, but they are very good codes when $p\equiv 3 \pmod 8$.  Computational results so far all exceed the Gilbert-Varshamov bound.  Since QQR codes are rate half codes, exceeding the Gilbert-Varshamov bound is equivalent to $\delta > 0.11$.  This gives Goppa's Conjecture a serious challenge.

\section{Weight Polynomials of QQR codes}
\label{sec:weight_polynomials_of_qqr_codes}

In this section, we will show a new result on weight polynomials of QQR codes.

\begin{defn}
	The \emph{weight polynomial} of a code (or a subset of it) is 
	\[
		A_C: = \sum^{n}_{i=0} A_i x^{n-i} y^i
	\]
	where $A_i$ denotes the number of codewords with weight $i$ in $C$, and $n$ is the length of the code.
\end{defn}

When computing the weight polynomials of QQR codes, we found that they are divisible by $(x^2 + y^2)^m$, where $m$ is at least its minimum distance minus 1, and $m \equiv 3 \pmod 4$.  See Table \ref{tab1}.  We can also see that for these $p$'s, $\delta$ are all well above $0.11$.

\begin{table}[h]
	\centering
	\begin{tabular}{c c c c}
		$p$ & $d$ & $\delta$ & Divisible by  \\
		\hline
		3 &  2 & 0.33 & $(x^2 + y^2)^3$ \\
		11  & 6 & 0.27 & $(x^2 + y^2)^7$ \\
		19 & 8 &  0.21 & $(x^2 + y^2)^7$ \\
		43  & 14 & 0.16& $(x^2 + y^2)^{15}$ \\
		59 & 18 & 0.15 & $(x^2 + y^2)^{19}$ \\
		67 & 22 & 0.16 & $(x^2 + y^2)^{23}$ \\
		\hline
	\end{tabular}
	\vspace{0.1 in}
	\caption{Computational Results}
	\label{tab1}
\end{table}

The fact that $m \equiv 3 \pmod 4$ can be shown using Gleason's theorem.
\begin{thm}[Gleason]
	\label{Gleason}
	If $C$ is self-dual code, then its weight polynomial $A_C$ is a polynomial in 
	\begin{align*}
		G(x,y) &=x^2 + y^2 \\
		J(x,y) &=x^2y^2(x^2-y^2)^2
	\end{align*}
\end{thm}

So for QQR codes, $A_C = \sum_{2i + 8j = 2p} a_i G(x,y)^i J(x,y)^j$.  Since the degree of $A_C$ is $2p \equiv 6 \pmod 8$, and $8j \equiv 0 \pmod 8$, we must have $i \equiv 3 \pmod 4$ for all $i$.  Therefore $m \equiv 3 \pmod 4$.

The other property of the weight polynomial is stated in the following theorem.

\begin{thm}
	\label{thm1}
	The weight polynomial of a QQR code is divisible by $(x^2 + y^2)^{d-1}$, where $d$ is its minimum distance.
\end{thm}

Since $m$ needs to satisfy $m \equiv 3 \pmod 4$, $m$ is larger than $d-1$ sometimes.

To prove Theorem \ref{thm1}, we need to introduce \emph{shadows}.

\subsection{Shadows}
\label{sub:shadows_of_codes}

We say a binary code is \emph{doubly-even} if all its weights are divisible by $4$, or \emph{singly-even} if its weights are even, but not doubly-even.

Let $C$ be a binary self-orthogonal code.  Then $C$ is even.  Let $C_0$ be the subset of doubly-even codewords of $C$.  If $C$ is singly-even, then $C_0$ is a linear subcode of index $2$ in $C$.

\begin{defn}[Shadow\cite{Rains2002}]
	The \emph{shadow} $S$ of a self-orthogonal binary code $C$ is 
	\[
		S = 
		\begin{cases}
			C_0^\perp \backslash C^\perp & \text{ if } C \text{ is singly-even}\\
			C^\perp & \text{ if } C \text{ is doubly-even}
		\end{cases}	
	\]
\end{defn}

We will follow the notation of \cite{Rains2002} and denote the weight polynomial of the shadow of $C$ as $S_C(x,y)$.  $S_C(x,y)$ can be computed from the weight polynomial of $C$.

\begin{lem}
	$S_C(x,y) = \frac{1}{|C|}A_C(x+y,i(x-y))$, where $i^2 = -1$.
	\label{lem3}
\end{lem}

We include the proof given in \cite{Rains2002} here since it's short.

\begin{proof}
	If $C$ is singly-even, this is immediate using MacWilliams identity.  Assume $C$ is doubly-even.
	$A_{C_0}$ consists of the terms in $A_C$ whose powers of $y$ are divisible by $4$. So 
	\[
		A_{C_0}(x,y) = \frac{1}{2}(A_C(x,y) + A_C(x,iy))
	\]

	Using MacWilliams identity, we have
	\[
		A_{C_0^\perp}(x,y) = \frac{1}{|C|}(A_C(x+y,x-y) + A_C(x+y, i(x-y)))
	\]

	So
	\[
		S_C(x,y) = A_{C_0^\perp} - A_{C^\perp} = \frac{1}{|C|}A_C(x+y, i(x-y))
	\]
\end{proof}

Under a simple change of variable, the following lemma is immediate.

\begin{lem}
	$A_C(x,y) = \frac{1}{|C|}S_C(x-iy,x+iy)$.
\end{lem}

When $C$ is singly-even, the shadow $S$ of $C$ is $C_0^\perp \backslash C^\perp$, which does not contain the $0$ codeword.  Therefore if $S$ has minimum distance $d$, $S_C$ is divisible by $(xy)^d$.  From the lemma above, we immediately have the following.

\begin{lem}
	Let $C$ be singly-even, and $d$ be the minimum distance of its shadow.  The weight polynomial of $C$ is divisible by $(x^2 + y^2)^d$.
	\label{lem2}
\end{lem}

Therefore, it is clear that, to prove Theorem \ref{thm1}, we just need to show the minimum distance of the shadow is at least its minimum distance minus 1.

In the next section, we will show this by proving a result about the automorphism group of extended QQR codes.

\subsection{Automorphism groups}
\label{sub:automorphism_groups_}

Let $C$ be the QQR code associated with $p$.  

Since $(r_Q,r_N) \in C$, and $wt(r_Q,r_N) = |Q| + |N| = p-1 \equiv 2 \pmod 4$, $C$ is singly-even.  By definition, the shadow of $C$ is $C_0^\perp \backslash C^\perp$.

Let $e_i = x^i(r_Q, r_N)=(x^i r_Q, x^i r_N)$.  

\[
	\{ e_i| i=0, \cdots, p-1 \}
\]
is a basis for $C$.

Denote $ \mathfrak{1}$ to be the all one codeword $( \sum_{i=0}^{p-1} x^i, \sum_{i=0}^{p-1} x^i)$.  

All the codewords can be expressed in both vector forms and polynomial forms like $\mathfrak{1}$.  We will alternate between them depending on which one is appropriate in the context.

\begin{prop}
	\label{prop1}
	$C_0$ is generated by 
	\[
		\{ \mathfrak{1} - e_i | i=0, \cdots, p-1\}
	\]
\end{prop}

The proof uses the following lemma.
\begin{lem}
	If each generator of a code has weight divisible by $4$, then so does every codeword.
	\label{lem1}
\end{lem}

This is a standard result that is easy to prove, and can be found in \cite{Macwilliams1977}.

\begin{proof}
	(of Proposition \ref{prop1}):

	Let $C'$ be the code generated by $\{ \mathfrak{1} - e_i |i=0, \cdots, p-1 \}$.  Since $ wt(\mathfrak{1} - e_i) = 2p - wt(e_i) = 2p - (p-1) \equiv 0 \pmod 4$, by Lemma \ref{lem1}, $C'$ is doubly-even.  Therefore $C' \subseteq C_0$.  

	Note that 
	\begin{align*}
		\sum_{i=0}^{p-1} \mathfrak{1} - e_i &=p \cdot \mathfrak{1} - \sum_0^{p-1} e_i \\
		&= \mathfrak{1} - (\sum_{i=0}^{p-1} x^i) e_0 \\
		&= \mathfrak{1} - (\sum_{i=0}^{p-1} x_i , \sum_{i=0}^{p-1} x_i) \\
		&= 0
	\end{align*}

	So $\{ \mathfrak{1} - e_i |i=0, \cdots, p-1 \}$ are linearly dependent.  The rank of $C'$ is less than or equal to $p-1$.

	On the other hand, if a subset of $\{1-e_i|i=0, \cdots, p-1\}$ with $k \le p-2$ elements is linearly dependent, we have
	\[
		\sum_{i=1}^k \mathfrak{1} - e_{n_i} = 0
	\]
	So
	\[
		\sum_{i=1}^k \mathfrak{1} - \sum_{i=1}^k e_{n_i} = 0
	\]

	If $k$ is odd, then $\sum_{i=1}^k e_{n_i}=\mathfrak{1}$.  But we already have $\sum_0^{p-1} e_i = \mathfrak{1}$.  Since $\{e_i\}$ is a basis of $C$, there can't be two different ways to write a vector in linear combinations of $e_i$'s.  Contradiction.

	If $k$ is even, then $\sum_{i=1}^k e_{n_i}=0$, contradictory to the $e_i$'s being linearly independent.

	We conclude that $C'$ has rank $p-1$.  Since $C_0$ also has dimension $p-1$, $C'=C_0$.  
\end{proof}

Let $ \alpha = (\sum_{i=0}^{p-1} x^i, 0)$, and $\beta = (0,\sum_{i=0}^{p-1} x^i)$.

\begin{prop}
	$C_0^\perp$ is generated by $C$ and $\alpha$(or by $C$ and $\beta$).
\end{prop}

\begin{proof}
	Since $C$ is self-dual, and $C_0 \subset C$, so $C \subseteq C_0^\perp$.  $\alpha$ is also in $C_0^\perp$ because
	\begin{align*}
		\alpha \cdot (\mathfrak{1}-e_i) & = \alpha \cdot \mathfrak{1} - \alpha \cdot e_i \\
		& = p - wt(x^i r_Q) \\
		& = p - |Q| \\
		& = (p + 1)/2 \\
		& = 0
	\end{align*}

	So $ \alpha \in C_0^\perp$.  Since $\alpha$ has odd weight, $\alpha \notin C$.  The code generated by $C$ and $\alpha$ has rank $p+1$, and so does $C_0^\perp$, and hence they are the same.
\end{proof}

From this proposition, we can see that $C$ is the even weight subcode of $C_0^\perp$.

Next, we will define an extended code for $C_0^\perp$ by adding two parity check columns.

\begin{defn}
	Let $\hat{C}$ be the extended code of $C_0^\perp$ by adding a parity check for the first $p$ bits and a parity check for the last $p$ bits, i.e. if 
	\[(a_0, a_1, \cdots, a_{p-1}, b_0, \cdots, b_{p-1} )\in C_0^\perp\]
	then it extends to \[
		(a_0, a_1, \cdots, a_{p-1}, \sum_{i=0}^{p-1} a_i, b_0, \cdots, b_{p-1}, \sum_{i=0}^{p-1} b_i) \in \hat{C}
	\]
\end{defn}

\begin{notation}
	If $c \in C$, denote $\hat{c}$ as the corresponding codeword in the extended code.
\end{notation}

Clearly, $\{\widehat{e_i}|i=0, \cdots, p-1\} \cup \{\hat{\alpha}\}$(or $\{\widehat{e_i}|i=0, \cdots, p-1\}\cup \{\hat{\beta}\}$) constitutes a basis for $\hat{C}$.

If we use $\{\widehat{e_i}|i=0, \cdots, p-1\} \cup \{\hat{\alpha}\}$ as the basis, then the generating matrix for $\hat{C}$ can be written as 

\[
	\left[
		\begin{array}{c c c| c |c c c|c}
			& & & 1 & & & & 1 \\
			& G_Q & & \vdots & & G_N & &\vdots   \\
			& & & 1 & & & & 1 \\
			\hline
			1 & \hdots & 1 & 1 & 0 & \hdots & 0 & 1	
		\end{array}
	\right]
\]

\bigskip

The permutations that showed up in our results act on the left half and the right half of a codeword in the same way.  For simplicity, in the following theorem and its proof, we relabel the positions in a codeword by their original positions modulo $p$, starting from $0$.  By convention, we label the parity check positions by $\infty$.  So starting from left, the positions in a codeword would be called position $0, \cdots, p-1, \infty, 0, \cdots, p-1, \infty$.

Below is the main result on the automorphism group of $\hat{C}$.

\begin{thm}
	The automorphism group of $\hat{C}$ contains $PSL_2(p)$ as a subgroup.  Here $PSL_2(p)$ is generated by the three permutations
	\begin{align*}
		S: y &\mapsto y+1 \\
		V: y &\mapsto \rho^2 y \\
		T: y &\mapsto - \frac{1}{y}
	\end{align*}
	where $\rho$ is a primitive element of $ \mathbb{F}_p$.
	\label{thm2}
\end{thm}

When $p \equiv 3 \pmod 8$, we have shown that the code generated by $[G_Q|G_N]$ is the same as the code generated by $[I|G_Q]$.  Therefore $\hat{C}$ also entails a generating matrix as following.
\[
	\left[
		\begin{array}{c c c| c |c c c|c}
			& & & 1 & & & & 1 \\
			& I & & \vdots & & G_Q & &\vdots   \\
			& & & 1 & & & & 1 \\
			\hline
			1 & \hdots & 1 & 1 & 0 & \hdots & 0 & 1	
		\end{array}
	\right]
\]

\bigskip

This form has been extensively studied before, and is usually referred to as \emph{bordered double circulant codes}.  It has been shown that $PSL_2(p)$ acts on these codes using the generating matrices above in previous work, such as \cite{Gaborit2001} and \cite{Tjhai2006a}.  Our proof is an alternate to those when $p\equiv 3 \pmod 8$.  When $p \equiv 7 \pmod 8$, these two codes are not equivalent, and therefore this is a new result.

The calculations presented in this proof are inspired by the proof of the theorem that the automorphism group of the extended quadratic residue code contains $PSL_2(p)$.  One can check \cite{Macwilliams1977} for that.  It uses the following theorem from number theory.  

\begin{thm}[Perron]
	Let $p = 4k + 3$, and let $Q$ be the quadratic residues in $ \mathbb{F}_p$, $N$ be the quadratic non-residues in $ \mathbb{F}_p$.  $a \neq 0 \in \mathbb{F}_p$.
	\begin{itemize}
		\item If $a \in Q$, then $\{a + r|r\in Q\}$ contains $ k$ quadratic residues and $k+1$ quadratic non-residues.
		\item If $a \in Q$, then $\{a + s|s\in N\}$ contains $0$, $ k$ quadratic residues and $k$ quadratic non-residues.
		\item If $a \in N$, then $\{a + r|r\in Q\}$ contains $0$, $ k$ quadratic residues and $k$ quadratic non-residues.
		\item If $a \in N$, then $\{a + s|s\in N\}$ contains $k+1$ quadratic residues and $k$ quadratic non-residues.
	\end{itemize}
	\label{thm3}
\end{thm}

\begin{proof}
	(of Theorem \ref{thm2})

	Let $p = 4k + 3$.

	$S$ sends position $i$ to $i + 1 (i=0, \cdots, p-1)$ and sends $\infty$ to $\infty$.  Since $C$ is double-circulant, it's fixed by $S$.  $\hat{\alpha}$ is also fixed by $S$.  

	$V$ fixes $C$ since it fixes both $r_Q$ and $r_N$.  The $\infty$ positions are sent to themselves.  $V$ also fixes $\hat{\alpha}$.

	Therefore, what's left to show is that $T$ also fixes $\hat{C}$.	

	We will show that by proving the following:
	\begin{itemize}
		\item $\widehat{e_0}^T = \mathfrak{1} + \widehat{e_0}$
		\item If $i \in Q$, $\widehat{e_i}^T = \hat{\beta} + \widehat{e_0} + \widehat{e_{\text{-}\frac{1}{i}}} $
		\item If $i \in N$, $\widehat{e_i}^T = \hat{\alpha} + \widehat{e_0} + \widehat{e_{\text{-}\frac{1}{i} }}$
	\end{itemize}

	\begin{itemize}

		\item $ \widehat{e_0}^T = \mathfrak{1} + \widehat{e_0}$:  If $y$ is a quadratic residue, then $ -\frac{1}{y}$ is a quadratic non-residue, and vice versa.  It follows immediately that the equality is true for all positions that are neither $0$ or $\infty$.  It's easy to check the equality also follows through at $0$ and $\infty$.

		\item If $i \in Q$, we will prove 
			\[
				\widehat{e_i}^T + \widehat{e_{\text{-} \frac{1}{i}}} = \hat{\beta} + \widehat{e_0}
			\]
			instead.

			Focus on the left $p+1$ bits first, and consider 
			\begin{equation}
				T( \sum_{r \in Q} x^{i + r}|1 ) + (\sum_{ r \in Q} x^{r - \frac{1}{i}}|1)
			\label{eqn2}
		\end{equation}

		According to Theorem \ref{thm3}, $\{ i+r | r\in Q\}$ has $k+1$ quadratic non-residues, $k$ quadratic residues.  Therefore $ - \frac{1}{i+r}$ has $k+1$ quadratic residues, and $k$ quadratic non-residues.  

		$\{ r - \frac{1}{i}| r \in Q\}$ contains $0$, $k$ quadratic residues, and $k$ quadratic non-residues.

		We want to know whether any terms in $T(\sum_{r \in Q} x^{i+r})$ would cancel with terms in $\sum_{ r \in Q} x^{ r - \frac{1}{i}}$.  In this case, they need to have the same powers of $x$.

		If $ -\frac{1}{i+r} = r' - \frac{1}{i}$, then 
		\begin{equation}
			r' = \frac{1}{i}-\frac{1}{i+r} = r \cdot (-\frac{1}{i+r} )\cdot (-\frac{1}{i})
			\label{eqn1}
		\end{equation}

		\begin{enumerate}
			\item If $-\frac{1}{i+r}$ is a quadratic residue, then (\ref{eqn1}) $\in N$.  Therefore there does not exist $r' \in Q$, s.t. $ - \frac{1}{i+r} = r' - \frac{1}{i}$.  Since there are $2k+1$ terms in (\ref{eqn1}) with quadratic residue powers of $x$, all terms with quadratic residue powers will show up in the sum.
			\item If $ - \frac{1}{i+r}$ is a quadratic non-residue, then there exists $r' \in Q$ satisfying $ - \frac{1}{i+r} = r' - \frac{1}{i}$.  Since there are $k$ quadratic non-residues in both $\{ i+r | r \in Q \}$ and $\{ r - \frac{1}{i}| r \in Q\}$, they will cancel in pairs.  None of the terms with quadratic non-residue powers of $x$ will show up in the sum.
		\end{enumerate}

		Lastly, check the $0$ and $\infty$ positions separately. 

		Since $T( \sum_{r \in Q} x^{ i+r}|1)$ has $1$ at position $0$, and so does $(\sum_{ r \in Q}  x^{ r- \frac{1}{i}}| 1 )$, they will cancel in the sum.

	$ T (\sum_{ r \in Q} x^{ i + r} | 1)$ has $0$ at $ \infty$, therefore the sum has $1$ at $\infty$.

	We conclude that
	\[
		T(\sum_{ r \in Q} x^{i+r}|1) + ( \sum_{ r \in Q} x^{ r - \frac{1}{i}}|1) = (\sum_{ r \in Q }x^r | 1)
	\]

	Now for the right $p+1$ bits, consider 
	\begin{equation}
		T( \sum_{ s \in N} x^{i+s} |1) + (\sum_{s \in N} x^{s - \frac{1}{i}}|1)
		\label{eqn4}
	\end{equation}		

	$\{ s+i | s \in n\}$ contains $0$, $k$ quadratic residues, $k$ quadratic non-residues.  Therefore $\{ - \frac{1}{s+i}\}$ contains $k$ residues and $k$ non-residues and $\infty$.

	$ \{ s - \frac{1}{i}\}$ contains $k+1$ quadratic residues, $k$ quadratic non-residues.  

	If $ - \frac{1}{i+s} = s'-\frac{1}{i}$, then 

	\begin{equation}
		s' = -\frac{1}{i+s} + \frac{1}{i} = s \cdot (-\frac{1}{i+s}) \cdot (-\frac{1}{i}) 
		\label{eqn5}
	\end{equation}

	\begin{enumerate}
		\item If $ - \frac{1}{i+s} \in Q$, (\ref{eqn5}) $\in Q$, there does not exist $s' \in N$ such that the two terms cancel.  Since there are in total $2k +1$ terms in (\ref{eqn4}) with quadratic residue powers, all terms with quadratic residue powers will show up in the sum.
		\item If $ - \frac{1}{i+s} \in N$, there exists $s' \in N$, such that $ - \frac{1}{i+s}=s'-\frac{1}{i}$.  Since there are $k$ quadratic non-residues in both $\{ i + s|s \in N\}$ and $\{ s - \frac{1}{i}\}$, they will cancel in pairs.  None of the terms with quadratic non-residue powers of $x$ will show up in the sum.
	\end{enumerate}

	Just like before, we can check (\ref{eqn4}) has $1$ at position $0$ and $0$ at $\infty$.

	We conclude that
	\[
		T(\sum_{s \in N} x^{i + s}|1) + (\sum_{ s \in N} x^{s - \frac{1}{i}}|1) = (\sum_{ r \in Q} x^r|0) 
	\]

	Combining these two parts, we have 
	\[
		\widehat{e_i}^T + \widehat{e_{\text{-}\frac{1}{i}}}= \widehat{ e_0} + \hat{\beta} 
	\]

\item The case of $i \in N$ can be proved in a similar fashion.

\end{itemize}

Lastly, $ \hat{\alpha}^T = \hat{\alpha}$.  

Since $T$ sends all basis elements into $\hat{C}$, $T$ fixes C.  
\end{proof}

In the same way that the result on automorphism groups of extended quadratic residue codes reveals the relation between its minimum distance and that of its expurgated code, this result leads to the following theorem on the minimum distance of the QQR code.

\begin{thm}
	Let $C$ be a QQR code.  The minimum distance of the shadow of $C$ is at least that of $C$ less 1.
\end{thm}

\begin{proof}

	If $d({C_0^\perp})$ is even, then since $C$ is the even weight subcode of $C_0^\perp$, we must have $d(C_0^\perp) = d(C)$.  Therefore $d(C_0^\perp \backslash C) > d(C)$.

	Let $d({C_0^\perp})$ be odd.  Let $c$ be a codeword in ${C_0^\perp}$ that achieves the minimum distance.  WLOG, assume $c$ has an odd number of non-zero elements in the first $p$ bits; then $c$ has an even number of non-zero elements in the last $p$ bits.  $\hat{c}$ is in the form

	\[
		( * \cdots * |1|* \cdots * |0|)
	\]

	We claim that we can find a position $y(0 \le y \le p-1)$, such that the coordinate on position $y$ from the left $p+1$ bits is $0$, and the coordinate on position $y$ from the right $p+1$ bits is also $0$.  Otherwise, for each position $y$, at least one of the coordinates is $1$, and so $wt(c) \ge p$, which contradicts the fact that $c$ has minimum weight.  So $c$ is in the  following form:

	\[
		(* \cdots 0 \cdots * |1| * \cdots 0 \cdots * |0|)
	\]

	Since $PSL_2(p)$ acts transitively on $\hat{C}$, we can find an element in $PSL_2(p)$ that exchanges $y$ and $\infty$.  Recall that $PSL_2(p)$ acts on the left half and the right half of a codeword in the same fashion.  We would therefore obtain a new codeword in $\hat{C}$ in the following form:

	\[
		(* \cdots 1 \cdots * |0| * \cdots 0 \cdots * |0|)
	\]

	By losing the two parity checks, we obtain a new codeword in $C_0^\perp$ that has weight $d(C_0^\perp) + 1$.  Note that this codeword also belongs to $C$, and therefore 
	\[
		d(C) \le d(C_0^\perp) + 1
	\]

	Equivalently,
	\[
		d(C_0^\perp\backslash C) =d(C_0^\perp) \ge d(C) -1
	\]

\end{proof}

Combining this with Lemma \ref{lem2}, we have provided a proof for Theorem \ref{thm1}.

\subsection{Computation algorithms for QQR codes}
\label{sub:computation_algorithms}

Computation of the weight polynomials is always an important topic in coding theory.  Researchers have come up with clever enumeration methods to reduce the computation load and speed up the process.  However in general, little was known about the structure of the weight polynomials, and therefore good tests for  computational results were missing. 

Theorem \ref{thm1} imposes a strong condition on the weight polynomials of QQR codes, and could serve as a test for existing and future computational results on the weight polynomials of QQR codes.  On the other hand, we can also use this to derive an algorithm around this and dramatically reduce the amount of computation needed.   

Since QQR codes are self-dual, by Gleason's theorem, their weight polynomials can be written as linear combinations of $G(x,y)^iJ(x,y)^j$, with $2i + 8j = 2p$.

Now for a QQR code $C$, let
\[
	A_{C} = \sum_{k = 0}^{2p} A_k x^{2p-k}y^k
\]

We should have 
\begin{equation}
	\sum_{k = 0}^{2p} A_i x^{2p-k}y^k = \sum_{2i + 8j = 2p} a_i G(x,y)^i J(x,y)^j
	\label{eqn13}
\end{equation}

We can use a recursive algorithm to recover the whole weight polynomial by knowing only a few $A_i$.

\begin{enumerate}
	\item  $A_0=1$ because of the 0 codeword.  Comparing coefficients of $x^{2p}$ on both sides of (\ref{eqn4}), we have $a_p = 1$.
	\item We obtain a new equation by subtracting $a_pG(x,y)^p$ on both sides of (\ref{eqn13})
		\begin{equation}
			W_C(x,y) - a_pG(x,y)^p = a_{p-4}G(x,y)^{p-4}J(x,y) + \ldots
			\label{eqn14}
		\end{equation}
		with the highest power of $x$ being $x^{2p-2}y^2$.
	\item $A_2 = 0$ because $d>2$.  Compare coefficients of $x^{2p-2}y^2$ on both sides of (\ref{eqn14}), and we have $a_{p-4}=-p$.
	\item  Repeat the steps until we have all the $a_i$'s.
\end{enumerate}

Since $A_C$ is divisible by $(x^2 + y^2)^{d-1}$, we only need $a_i$ for $i \geq d-1$. 

For the case of $p=59$, computing the whole weight enumerator using MAGMA using brutal force would take 190 years.  Using this strategy, however, we need only a few $A_i$'s.  Assume the minimum distance is at least 14, which is reasonable based on the result for $p=43$.  This can also be confirmed computationally.  Then the weight polynomial is divisible by $(x^2 + y^2)^{15}$. Therefore to get all the $a_i(i=15, \cdots, 59)$, we only need $A_j$'s for $j= 14, 16, 18, 20, 22$, which takes a few hours to compute.  Note that this time is based on using existing commands in MAGMA, and could be even faster if combined with enumeration techniques.

The huge speed up is because not all coefficients are created equal.  The ones we needed for our computation are those $A_i$'s with very small or very large $i$.  These take much less time than those ones in the middle.  We are avoiding, and computing using our algorithm, those coefficients in the middle that could take years to compute.

\subsection{Zeta polynomials and Riemann hypothesis}
\label{sub:zeta_polynomials_and_riemann_hypothesis}

In the last part of this section, we answer a question that was originally posted by Joyner in \cite{Joyner2008}.

Let $d$ be the minimum distance of a code $C$ and $d^\perp$ the minimum distance of its dual code.  Iwan Duursma introduced the zeta function $Z=Z_C$ associated to a linear code $C$ over a finite field $\mathbb{F}_q$ \cite{Duursma2001}.

\[
	Z(T) = \frac{P_C(T)}{(1-T)(1-qT)}
\]
where $P_C(T)$ is a polynomial of degree $n+2-d-d^\bot$.  This is a polynomial with rational coefficients, called the zeta polynomial of the code $C$.

Given a self-dual code, it is always of interest whether its zeta polynomial satisfies the Riemann hypothesis.  (In other words, its roots occur in self-reciprocal pairs).  Joyner asked this question about the QQR codes for $p \equiv 3 \pmod 4$.  Using SAGE to compute zeta polynomials, we found that it does not satisfy the Riemann hypothesis for $p=23$.  

For $p=23$, it has 15 pairs of complex conjugate roots of absolute value $\frac{1}{\sqrt{2}}$, together with real roots 0.508887881 and 0.982534697.  The last two roots cause the code to fail the Riemann hypothesis.

\subsection{Quadratic residue codes}
\label{sub:quadratic_residue_codes}

As mentioned earlier, when $p \equiv 7 \pmod 8$, QQR codes have a close relation with quadratic residue codes.  In this section, we will first prove a similar divisibility property of quadratic residue codes using shadows of codes, and use their relation with QQR codes to give an alternative proof to Theorem \ref{thm1}.

We first introduce expurgated quadratic residue codes.

\begin{defn}[Expurgated quadratic residue code]
	Let $Q$ and $N$ be the quadratic residue codes associated with $p$.  

	The even subcodes of $Q$ and $N$, which are denoted as $\bar{Q}$ and $\bar{N}$ respectively, are called \emph{expurgated quadratic residue codes}.

\end{defn}

We list some well-known properties of quadratic residue codes that will be used later.  All can be found in \cite{Macwilliams1977}.

\begin{prop}[\cite{Macwilliams1977}] 
	Let $p\equiv \pm 1 \pmod 8$, then
	\hfill
	\begin{enumerate}
		\item $Q$ and $N$ both have dimension $ \frac{1}{2}(p+1)$.  $\bar{Q}$ and $\bar{N}$ have dimension $ \frac{1}{2}(p-1)$.
		\item If $p \equiv 3 \pmod 4$, $Q^\perp = \bar{Q}$, $N^\perp = \bar{N}$.  If $p \equiv 1 \pmod 4$, $Q^\perp = \bar{N}$, $N^\perp=\bar{Q}$.
		\item $Q$ is generated by $\bar{Q}$ and the all one codeword, $N$ is generated by $\bar{N}$ and the all one codeword.
		\item $\bar{Q}$ and $\bar{N}$ are doubly-even.
		\item Let $d$ be the minimum distance.  If $p \equiv -1 \pmod 8$, $d \equiv 3 \pmod 4$.  If $p \equiv 1 \pmod 8$, then $d$ is odd.
	\end{enumerate}		
	\label{prop7}
\end{prop}

Similar to Lemma \ref{lem4}, we can prove the following.

\begin{lem}
	\label{lem4}
	Let $C$ be a self-orthogonal binary code, and $l$ its maximum weight.  Then the weight polynomial of its shadow is divisible by $(x+y)^{n-l}$.
\end{lem}

\begin{proof}
	\[
		A_C = x^n + \cdots + x^{n-l}y^l
	\]
	So $A_C$ is divisible by $x^{n-l}$.

	From Lemma \ref{lem3}, it's immediate that $S_C$ is divisible by $(x+y)^{n-l}$.
\end{proof}

We can now prove a divisibility property on the weight polynomial of the quadratic residue code.

\begin{thm}
	Let $p \equiv \pm 1 \pmod 8$ be prime, and $C$ the quadratic residue code associated with $p$, then $A_C$ is divisible by $(x+y)^d$, where $d$ is its minimum distance.
	\label{thm4}
\end{thm}

\begin{proof}
	We will only prove this for $p \equiv -1 \pmod 8$.  The case for $p \equiv 1 \pmod 8$ is similar.

	When $p \equiv -1 \pmod 8$, let $\bar{C}$ be the corresponding expurgated quadratic residue code.  Then $\bar{C}^\perp = C$ and $C$ is generated by $\bar{C}$ and the all one codeword.

	Since $p \equiv 1 \pmod 8$, $d \equiv 3 \pmod 4$.  Let $c$ be a codeword that achieves the minimum distance, and let $c'$ be the sum of $c$ and the all one codeword.  Then $c'$ has even weight $p-d$, and therefore is contained in $\bar{C}$.  Since $C$ has odd length $p$, $\bar{C}$ does not contained the all one codeword, and therefore $c'$ has the largest weight in $\bar{C}$.

	By definition, $C$ is the shadow of $\bar{C}$.  Therefore by Lemma \ref{lem4}, the weight polynomial of $\bar{C}$ is divisible by $(x+y)^d$.
\end{proof}

When $p\equiv 7 \pmod 8$, the QQR code $C$ is the even subcode of $Q \oplus N$.  Therefore we have the following relation between their weight polynomials.

\begin{prop}
	When $p \equiv 7 \pmod 8$, let $C$ be the QQR code associated with $p$, and $Q$ the corresponding quadratic residue code, then
	\[
		A_C = (A_Q^2(x,y) + A_Q^2(x,-y))/2
	\]
	\label{prop8}
\end{prop}
\begin{proof}

	The weight polynomial of the direct sum of two codes is the product of their respective weight polynomials, and the weight polynomial of the even weight subcode of a code is just the sum of terms in its weight polynomials with even powers of $y$.  

	This proposition is immediate after combining these two facts with Proposition \ref{prop9}.
\end{proof}

We now give an alternative proof to Theorem \ref{thm1} in the case $p \equiv 7 \pmod 8$ using this relation.

\begin{proof}
	When $p \equiv 7 \pmod 8$, let $C$ be the QQR code associated with $p$, and let $Q$ and $\bar{Q}$ be the corresponding quadratic residue code and the expurgated quadratic residue code respectively.

	Since $Q$ is generated by $\bar{Q}$ and the all one codeword,
	\[
		A_Q = A_{\bar{Q}}(x,y) + A_{\bar{Q}}(y,x)
	\]

	Note that the minimum distance of $Q$ is $d-1$, we have	\[
		(x + y)^{d-1} | A_{\bar{Q}}(x,y) + A_{\bar{Q}}(y,x)
	\]

	Change $y$ for $iy$ and $-iy$ respectively, where $i^2=-1$.  We obtain
	\begin{align*}
		(x + iy)^{d-1} &| A_{\bar{Q}}(x,iy) + A_{\bar{Q}}(iy,x) \\
		(x -i y)^{d-1} &| A_{\bar{Q}}(x,-iy) + A_{\bar{Q}}(-iy,x)
	\end{align*}

	Since $\bar{Q}$ is doubly-even, for each term in $A_Q$, powers of $y$ are all divisible by $4$ and powers of $x$ are $3 \pmod 4$.  

	Therefore
	\begin{align*}
		A_{\bar{Q}}(x,iy) = A_{\bar{Q}}(x,y) ,&\: A_{\bar{Q}}(iy,x) = -iA_{\bar{Q}}(y,x)\\
		A_{\bar{Q}}(x,-iy) = A_{\bar{Q}}(x,y),&\: A_{\bar{Q}}(-iy,x) = iA_{\bar{Q}}(y,x)
	\end{align*}

	Hence 
	\begin{align*}
		&(x + iy)^{d-1}|A_{\bar{Q}}(x,y) - iA_{\bar{Q}}(y,x)\\
		&(x -iy)^{d-1}|A_{\bar{Q}}(x,y) + i A_{\bar{Q}}(y,x)
	\end{align*}

	Lastly 
	\begin{align*}
		A_C &= (A_Q(x,y)^2 + A_Q(x,-y)^2)/2 \\
		&= ((A_{\bar{Q}}(x,y) + A_{\bar{Q}}(y,x))^2 + (A_{\bar{Q}}(x,-y) + A_{\bar{Q}}(-y,x))^2)/2 \\
		&=((A_{\bar{Q}}(x,y) + A_{\bar{Q}}(y,x))^2 + (A_{\bar{Q}}(x,y) - A_{\bar{Q}}(y,x))^2)/2 \\
		&=A_{\bar{Q}}(x,y)^2 + A_{\bar{Q}}(y,x)^2 \\
		&= (A_{\bar{Q}}(x,y) + i A_{\bar{Q}}(y,x))(A_{\bar{Q}}(x,y) - i A_{\bar{Q}}(y,x))
	\end{align*}

	which is divisible by $(x^2 + y^2)^{d-1}$.
\end{proof}

\subsection{Weight polynomials of quadratic residue codes in the literature}
\label{sub:weight_polynomial_corrections_for_quadratic_residue_codes_when_p_113_127_}

Previously, weight polynomials of quadratic residue codes have been computed up to $p= 167$.  We are referring to the online table \emph{Weight Distributions of Quadratic Residue and Quadratic Double Circulant Codes over GF(2)}\cite{Tjhai}.  This table is also the source for the same entries on \emph{The On-Line Encyclopedia of Integer Sequences} (OEIS)\cite{Sloane}.

We tested these results against Theorem \ref{thm4}.  The results are shown in Table \ref{tab2}.

\begin{table}[h]
	\centering
	\begin{tabular}{c c c c}
		$p$ & $k$ & $d$ & Divisible by  \\
		\hline
		89 & 45  & 17 & $(x+y)^{17}$ \\
		97 & 49  & 15 & $(x+y)^{15}$ \\
		103 & 52 & 19 & $(x+y)^{19}$ \\
		113 & 57 & 15 & $\color{red}(x+y)$ \\
		127 & 64 & 19 & $\color{red}(x+y)$ \\
		137 & 69 & 21 & $\color{red}(x+y)$ \\
		151 & 76 & 19 & $\color{red}(x+y) $ \\
		167 & 84 & 23 & $\color{red}(x+y)$ \\
		\hline
	\end{tabular}
	\vspace{0.1 in}
	\caption{Weight polynomials posted on \cite{Tjhai}}
	\label{tab2}
\end{table}

The weight polynomials posted for $p=113,127,137,151,167$ are only divisible by $x+y$ and no further, and therefore errors existed in these results.  We investigated each case and give the results as follows.

\subsubsection{$p=137, 151,167$}
\label{ssub:p_137_151_167_}

For $p=137,151,167$, we found that the numbers from the original references of the online table are different from the numbers posted in the online table.  

In particular, for $p=137$, the numbers in the paper \cite{Tjhai2008} are different from the online table.  

For $p= 151, 167$, the numbers in \cite{Tomlinson2017} are different from the online table.  

We tested Theorem \ref{thm4} against the numbers in these references and confirmed they satisfy the divisibility conditions.  See Table \ref{tab4}.

\begin{table}[h]
	\centering
	\begin{tabular}{c c c c}
		$p$ & $k$ & $d$ & Divisible by  \\
		\hline
		137 & 69 & 21 & $(x+y)^{21}$ \\
		151 & 76 & 19 & $(x+y)^{19}$ \\
		167 & 84 & 23 & $(x+y)^{23}$ \\
		\hline
	\end{tabular}
	\vspace{0.1 in}
	\caption{Weight polynomials in references}
	\label{tab4} 
\end{table}

They also satisfy the following checks:
\begin{enumerate}
	\item All the $A_i$'s are divisible by $p$, except for $A_0$ and $A_p$.  (This should hold because quadratic residue codes are cyclic.)
	\item $\sum_0^p A_i = 2^k$.
	\item $A(x,y)$ is divisible by $(x+y)^d$.
	\item The corresponding weight polynomial for the extended quadratic residue codes satisfy the MacWilliams identity.  (This should hold because extended quadratic residue codes are self-dual).  In other words, 
		\[
			f(x,y) = x(A(x,y) + A(x,-y))/2 + y(A(x,y) - A(x,-y))/2
		\]
		should satisfy
		\[
			f(x,y) = \frac{1}{2^{k}}f(x+y,x-y)
		\]
\end{enumerate}

Since the divisibility condition and the four checks are highly non-trivial, we believe the original references are correct, and the numbers in the online tables are off possibly due to rounding using double-precision floating-point format\cite{Contributors}.

\subsubsection{Correction for $p=113$}
\label{ssub:correction_for_p_113_}

We could not find a reference for these numbers, but were able to find the correct weight polynomial in this case.

In fact, all the numbers in the online table are correct except $A_{56}$ and $A_{57}$ should be changed from $10375431209297308$ to $10375431209297309$.  The resulted weight polynomial satisfy Theorem \ref{thm4} and the four checks we listed above.

\subsubsection{Correction for $p=127$}  
\label{ssub:correction_for_p_127_}

We could not find a reference for $p=127$ either.  Therefore we deduced the correct weight polynomial based on the criteria it needs to satisfy.

By Theorem \ref{thm4}, the weight polynomial $A_Q$ is divisible by $(x+y)^{19}$, therefore we have 
\begin{equation}
	(x+y)^{19} (\sum^{108}_{i=1} c_i x^{108-i}y^i) = \sum^{127}_{j=1} A_i x^{127-j} y^j
	\label{eqn12}
\end{equation}
for some integers $c_j$'s.  Our goal is to solve these $c_j$'s.

Expand Equation \ref{eqn12} and compare coefficients, we have 
\[
	\sum_{i+k=j} {19 \choose k} c_i = A_j
\]

Therefore if the number of correct $A_j$'s we know are greater than or equal to $108$, we can set up enough equations to solve all the $c_i$'s.  Below are the $A_j$'s we use.

\begin{itemize}
	\item Since $d = 19$, we know $A_0=A_{108}=1$ and $A_i=0(i=1,\cdots,18,109, \cdots 126)$.   
	\item $A_i=0$ when $i \equiv 1,2 \pmod 4$\cite{Macwilliams1977}.  
	\item We use $A_{19}$ to $A_{43}$ (and $A_{84}$ to $A_{108}$) from the posted result, hoping they were correct.  Therefore we used 13 coefficients from the table. 
\end{itemize}

Fortunately the 13 coefficients we use from the table seem correct.  The weight polynomial we obtained using this approach passes the four checks we listed above.  Therefore we are confident the answer is correct.

Below are the $A_j$'s that needed to be corrected.  Since the weight polynomials for quadratic residue codes are symmetric, we only listed $A_j$'s for $j \le 63$.

\begin{table}[h]
	\label{tab3}
	\centering
	\begin{tabular}{c r r }
		$i$ & $A_i$ in table & $A_i$ corrected\\
		\hline
		51 & 223367511592873280 & 223367511592873284 \\
		52 & 326460209251122496 & 326460209251122492 \\
		55 & 840260234424082176 & 840260234424082220 \\
		56 & 1080334587116677120 & 1080334587116677140 \\
		59 & 1899366974583683328 & 1899366974583683220 \\
	       	60 & 2152615904528174336 & 2152615904528174316 \\
		63 & 2596788489999036416 & 2596788489999036307 \\
		\hline
	\end{tabular}
	\vspace{0.1 in}
	\caption{Correction for $p=127$}
\end{table}

\section{Weight Distribution and Hyperelliptic Curves}
\label{sec:weight_distribution_and_hyperelliptic_curves}

The weights of QQR codes are closely linked with numbers of points on corresponding hyperelliptic curves.  This connection enables us to study the distribution of number of points on hyperelliptic curves using the weight distribution of QQR codes.  

In this section, we will first show a result on the weight distribution of QQR codes, and then demonstrate how to use this result to prove a corresponding result on hyperelliptic curves.

Let $S \subseteq \mathbb{F}_p$, and let $c=(r_Qr_S,r_Nr_S)$.  According to Proposition \ref{prop-Joyner}, we have 
\begin{itemize}
	\item If $|S|$ is even, then 
		\[
		wt(c) = 2p-|X_S( \mathbb{F}_p)|
	\]
\item If $|S|$ is odd, then
	\[
	wt(c) = |X_{S^c} (\mathbb{F}_p)|
\]	
\end{itemize}

\begin{rmk}
	The original statement posted in Joyner's paper is slightly different, since his count includes points at infinity.  For simplicity, we modified the statement to restrict only to affine points.
\end{rmk}

In order to link the weight distribution of the QQR codes and the point distributions of hyperelliptic curves, we also need the following results:

\begin{prop}
	Let $S \subseteq \mathbb{F}_p$. 
	\begin{itemize}
		\item If $|S|$ is even, then $|X_S( \mathbb{F}_p)| \equiv 2 \pmod 4$.
		\item If $|S|$ is odd, then $|X_S( \mathbb{F}_p)| \equiv 3 \pmod 4$.
	\end{itemize}	
	\label{prop3}
\end{prop}

We give a sketch of the proof as following:
\begin{proof}(sketch)
	Let $\chi = ( \frac{}{p})$ be the quadratic residue character, which is 1 on the quadratic residues $Q \in \mathbb{F}_p$, -1 on the quadratic non-residues, and 0 on 0.  

	Then
	\[
		|X_S( \mathbb{F}_p)| = p + \sum_{ a \in \mathbb{F}_p} \chi(f_S(a))
	\]

	Since $ p \equiv 3 \pmod 4$, we just need to show the following:

	\begin{itemize}
		\item If $|S|$ is even, $\sum_{ a \in \mathbb{F}_p} \chi(f_S(a)) \equiv 3 \pmod 4$.
		\item If $|S|$ is odd, $\sum_{ a \in \mathbb{F}_p} \chi(f_S(a)) \equiv 0 \pmod 4$.
	\end{itemize}

	These are proven by induction on $|S|$, as follows:

	\begin{enumerate}
		\item 	Consider the simplest case.  If $S = \{r \}$, then
			\[
				\sum_{ a \in \mathbb{F}_p} \chi( f_S(a)) = \sum_{ a \in \mathbb{F}_p} \chi( a - r) = \sum_{ a \in \mathbb{F}_p} \chi(a) = 0
			\]
		\item If $|S| > 1$, then take $s \in S$, and let $R = S \backslash \{s\}$.  We can show that 	
		\begin{equation}
			\sum_{ a \in \mathbb{F}_p} \chi(f_S(a))\equiv  - p + |S| -   \sum_{a \in \mathbb{F}_p} \chi(f_R(a)) + \chi(f_R(s)) + \sum_{ a \in R} \chi(a-s) \pmod 4
			\label{eqn11}
		\end{equation}

	\item In particular, when $|S|=2$, by (\ref{eqn11}), we have \[
			\sum_{a \in \mathbb{F}_p} \chi(f_S(a)) \equiv 3 \pmod 4
		\]

		\item Assume for $|S| < n$, the statements are true.  We can show the following relation
			\[
				\chi(f_R(s)) + \sum_{ a \in R} \chi(a-s) \equiv 1 - |R| \pmod 4
			\]

			Combining with (\ref{eqn11}), we have 
			\[
				\sum_{ a \in \mathbb{F}_p} \chi(f_S(a)) \equiv   3 - \sum_{ a \in \mathbb{F}_p} \chi(f_R(a)) 
			\]

			If $|S|$ is odd, then $|R|$ is even.  By assumption $\sum_{a \in \mathbb{F}_p} \chi(f_R(a)) \equiv 3 \pmod 4$, and therefore 
			\[
				\sum_{ a \in \mathbb{F}_p} \chi(f_S(a)) \equiv 0 \pmod 4
			\]

			Similarly, if $|S|$ is even, then $|R|$ is odd, and we have 
			\[
				\sum_{ a \in \mathbb{F}_p} \chi(f_S(a)) \equiv 3 \pmod 4
			\]
	\end{enumerate}
\end{proof}

In Proposition \ref{prop3}, we notice that when $|S|$ is even, a codeword $c$ is associated with a curve $y^2 = f_S(x)$.  When $|S|$ is odd, $c$ is associated with a curve $y^2 = f_{S^c} (x)$ with $|S^c|$ even. 

Therefore the curves that are linked with QQR codes are in the form 
\[
	y^2 = f_S(x)
\]
where $|S|$ is even.  We denote this set of curves as $ \mathscr{C}_p$. 

Let $B_k$ be the number of curves in $\mathscr{C}_p$ that have $k$ affine points over $ \mathbb{F}_p$, and let $A_k$ be the number of codewords with weight $k$.

From Proposition \ref{prop-Joyner}, it's clear that
\[
	A_k = B_k + B_{2p-k}
\]

and from Proposition \ref{prop3}, we have the following relation between $A_k$ and $B_k$:

\begin{prop}

	Let $A_k$ and $B_k$ be described as above, then

	\begin{itemize}
		\item If $k$ is odd, $A_k = 0$.
		\item If $k \equiv 0 \pmod 4$, $A_k = B_{2p-k}$.
		\item If $k \equiv 2 \pmod 4$, $A_k = B_k$.
	\end{itemize}

	\label{prop5}

\end{prop}

Therefore we have obtained an explicit relation between the weight distribution of a QQR code associated with $p$ and the point distribution of the hyperelliptic curves in the set $\mathscr{C}_p$.

The following diagram illustrates this interlacing pattern for $p=11$.  $A_k$'s can be obtained from $B_k$'s by symmetrizing the distribution of $B_k$'s with respect to $k$.

\bigskip
\begin{center}
	\begin{tikzcd}[column sep = 1.5ex,row sep = 3ex]
		k & 0 & 2 & 4 & 6 & 8 & 10 & 12 & 14 & 16 & 18 & 20 & 22 \\
		B_k & 0 & 0 & 0 & 77 \ar[d,color=red,dashed] \ar[rrrrrd,dashed,color=red] & 0 & 616 \ar[d,thin,color=green,dashed] \ar[rd,thin,color=green,dashed] & 0 & 330 \ar[d,thin,color=gray,dashed] \ar[llld,color=gray,dashed,thin] & 0 & 0 & 0 & 1 \ar[d,thin,color=blue,dashed] \ar[llllllllllld,thin,color=blue,dashed] \\
		A_k & 1 & 0 & 0 & 77 & 330 & 616 & 616 & 330 & 77 & 0 & 0 & 1
	\end{tikzcd}
\end{center}

\bigskip
Next we will show a result on the point weight distribution of QQR codes.  But firstly we need to formally define the moments of QQR codes from the discrete values of $A_k$.  These are standard definitions and can be found in \cite{Macwilliams1977}.

\begin{defn}[Moments]
	For a code $C$ of length $n$, let $a_j = A_j/2^k$, where $A_j$'s are the coefficients of its weight polynomial.  The \emph{mean} and \emph{variance} of $C$ are defined by 
	\begin{align*}
		\mu &= \sum^{n}_{j=0} j a_j \\
		\sigma^2 &= \sum^{n}_{j=0} (\mu-j)^2 a_j
	\end{align*}	
	and the $r^{th}$ central moment is
	\[
		\mu_r = \sum^{n}_{j=0} (\frac{\mu-j}{\sigma})^r a_j
	\]
\end{defn}

\begin{defn}[Cumulative Distribution Function]
	The \emph{cumulative distribution function}(c.d.f.) $A(z)$ of a code $C$ is given by 
	\[
		A(z) = \sum^{n}_{j \ge \mu - \sigma z} a_j
	\]
\end{defn}

The following is a classical theorem on the weight distribution of codes.

\begin{thm}[Sidel'nikov \cite{Macwilliams1977}]
	Let $C$ be a binary code, and $d^\perp \ge 3$ the minimum distance of its dual code $C^\perp$, then
	\[
		|A(z) - \Phi(z)| \le \frac{20}{\sqrt{d^\perp}}
	\]
	where $\Phi(z)$ is the c.d.f. for the normal distribution.
	\label{thm5}
\end{thm}

In other words, if $d^\perp$ tends to infinity for a series of codes, then its weight distribution is asymptotically normal. 

\begin{figure}[htpb]
	\centering
	\includegraphics[width=0.8\linewidth]{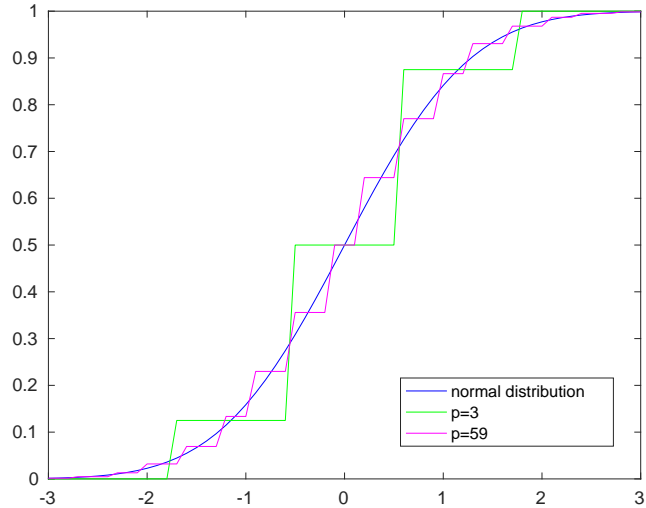}
	\caption{distribution comparison}
	\label{fig:distribution}
\end{figure}

In Proposition \ref{prop9}, we showed that when $p \equiv 7 \pmod 8$, the minimum distance of a QQR code is the minimum distance of its corresponding quadratic residue code plus 1.  Since the minimum distances of quadratic residue codes have a well-known lower bound of $\sqrt{p}$, the minimum distance of QQR codes are bounded below by $\sqrt{p} + 1$.

When $p \equiv 3 \pmod 8$, Helleseth and Voloch have proven the following bound for QQR codes.

\begin{thm}\cite{Helleseth2004}
	The minimum distance $d$ of a QQR code when $p \equiv 3 \pmod 8$ is bounded by 
	\[
		d \ge \frac{2(p+ \sqrt{p})}{\sqrt{p}+3}
	\]
\end{thm}

Combining these and Theorem \ref{thm5}, we have the following theorem.

\begin{thm}
	The weight distribution of QQR codes are asymptotically normal.
\end{thm}

Figure~\ref{fig:distribution} shows the comparison of the c.d.f among normal distribution, the QQR code with $p=3$, and the QQR code with $p=59$.  Both $p=3$ and $p=59$ are approximating the normal distribution.  Their c.d.f's are step functions by construction.  $p=59$ oscillates more frequently and more closely to the normal distribution.  We imagine that with $p$ bigger, the oscillation will become more frequent and closer to the normal distribution.    

Since proposition \ref{prop5} gives an explicit relation between the point distribution of hyperelliptic curves and weight distribution of QQR codes, we can get all the $A_k$'s by using $B_k$'s.  Namely, set

\[
	A_k = \begin{cases}
		B_k & k \equiv 2 \pmod 4 \\
		B_{2p-k} & k \equiv 0 \pmod 4 \\
		0 & \text{otherwise}
	\end{cases}
\]

Therefore we have shown that, after symmetrizing the point distribution of hyperelliptic curves in $ \mathscr{C}_p$, the result will converge to the normal distribution when $p \rightarrow \infty$.

A recent study by Larsen\cite{Larsen2008} showed that, more or less, for a random curve of random genus, over a random finite field $\mathbb{F}_q$, $T/\sqrt{q}$, is normally distributed, where $T$ is the number of points on the curve.  More precisely, 
\begin{itemize}
	\item Fix $g$.  As $q \rightarrow \infty$, $T/\sqrt{q}$ defines a measure $\mu_g$ on $[-2g,2g]$.  e.g. for $g=1$, $\mu_1$ is the Sato-Tate measure.
	\item The limit of these measures $\mu_g$ when $g \rightarrow \infty$ is the measure given by the standard normal distribution.
\end{itemize}

Our result is a variant of Larsen's.  The main differences are:
\begin{itemize}
	\item The set of curves in Larsen's result consist of all hyperelliptic curves defined over $ \mathbb{F}_q$ while ours is a subset of that given by the definition of $\mathscr{C}_p$.  
	\item We showed that after being symmetrized, the distribution approaches the standard normal distribution, while Larsen's result is on the point distribution itself.  
	\item Larsen's result uses theoretical results on hyperelliptic curves among others, while our result is simply a corollary from the study on QQR codes.
\end{itemize}

\section{Conclusion}
\label{sec:conclusionconclusion}

In this paper, we begin by reviewing some of the known properties of QQR codes, and proved that $PSL_2(p)$ acts on the extended QQR code when $p \equiv 3 \pmod 4$.  Using this discovery, we then showed their weight polynomials satisfy a strong divisibility condition, namely that they are divisible by $(x^2 + y^2)^{d-1}$, where $d$ is the corresponding minimum distance.  Using this result, we were able to construct an efficient algorithm to compute weight polynomials for QQR codes and correct errors in existing results on quadratic residue codes.

In the second half, we use the relation between the weight of codewords and the number of points on hyperelliptic curves to prove that the symmetrized distribution of a set of hyperelliptic curves is asymptotically normal.

\section{Acknowledgment}
\label{sec:acknowledgment}

The authors want to thank Prof. Iwan M Duursma for his keen observation on the relation between weight polynomials and shadows.

\bibliographystyle{AIMS}

\bibliography{Ref}

\end{document}